\documentclass[11pt,draftclsnofoot,peerreview]{IEEEtran}
\usepackage[linesnumbered,ruled,vlined]{algorithm2e}
\usepackage[english]{babel}
\usepackage{tikz}
\usepackage{xcolor}
\usepackage{amsmath, amssymb}
\usepackage{amsthm}
\usepackage{float}
\usepackage{subcaption}
\usepackage{enumerate}
\allowdisplaybreaks

\newcommand{\ones}[1]{{\lvert #1 \rvert_{1}}}

\newcommand{\s}{\text{c}}
\newcommand{\remove}[1]{}

\newtheorem{theorem}{Theorem}
\newtheorem{lemma}{Lemma}
\newtheorem{example}{Example}
\newtheorem{corollary}{Corollary}
\theoremstyle{definition}

\newcommand{\Section}[1]{\section{#1}}
\newcommand{\Subsection}[1]{\subsection{#1}}

\allowdisplaybreaks
\begin{document}
\title{Optimal Binary Variable-Length Codes with a Bounded Number of 1's per Codeword: Design, Analysis, and Applications\thanks{This is a full version of a paper submitted to ISIT 2025. This work was partially supported by project SERICS (PE00000014) under the NRRP MUR program funded by the EU-NGEU.}}
\author{\IEEEauthorblockN{Roberto Bruno, Roberto De Prisco, and Ugo Vaccaro \IEEEmembership{Senior Member,\ IEEE }} \\\IEEEauthorblockA{Department of Computer Science, University {of Salerno}, Fisciano (SA), Italy\\
Emails: \{rbruno, robdep, uvaccaro\}@unisa.it}}

\maketitle

\vspace*{-.5cm}
\begin{abstract}
In this paper, we consider the problem of constructing optimal average-length binary 
codes
under the constraint that each codeword must contain at most $D$ ones, where $D$ is 
a given input 
parameter.
We provide an $O(n^2D)$-time complexity algorithm for the construction of such codes, where $n$ is the number of codewords. We also describe several scenarios where the need to design  
these kinds of codes naturally arises. Our algorithms allow us to construct both optimal
average-length prefix binary codes and optimal
average-length alphabetic binary codes. In the former case, our $O(n^2D)$-time algorithm
substantially improves on the previously known $O(n^{2+D})$-time complexity algorithm for the
same problem.
We also provide a Kraft-like inequality for the existence of (optimal) variable-length binary codes, subject to the above-described constraint on the number of 1's in each codeword.

\end{abstract}


\Section{Introduction}\label{intr}
Variable-length codes find plenty of applications in data compression \cite{sal};
additionally,  variable-length codes are  deeply interconnected with 
search theory.
Specifically, any search procedure that identifies objects within a given space by sequentially performing appropriate tests inherently defines a variable-length encoding for the elements in that space. To illustrate, each possible test outcome can be represented using a distinct symbol from a finite alphabet. The concatenation of these encoded test outcomes, which ends in identifying an object 
$x$, provides a valid encoding of
$x$,  for any $x$
 in the search space. The relationships between variable-length codes and 
search procedures have been thoroughly investigated; the main results are 
presented in the  excellent monographies \cite{AW,aigner,Picard}.

Among  search procedures, 
the binary search stands out as an important algorithmic technique to solve
the computational  problem whereby an unknown element 
$x$ is located within a finite search space 
$S$, and the objective is to identify 
$x$ by performing queries that yield {\sc Yes/No} answers. These kind of queries typically fall into one of two categories: comparison queries or membership queries.
In the case of comparison queries, the unknown element 
$x$ resides in an ordered set that can be represented as 
$[n]=\{1, \ldots, n\}$
The queries then take the form: "Is the unknown number 
$x$ greater than 
$j$, where $j\in [n]$?" (see \cite{knuth} for an encyclopedic treatment of the subject).
For membership queries, the search space 
$S$ may be unordered. The queries are phrased as: "Does the unknown number 
$x$ belong to a subset $A \subseteq S$?"
This latter problem has also been extensively studied (see, for example, \cite{AW, aigner, mehlhorn}) and has significant connections to a variety of fields, including  Exact Learning \cite{Da}, Combinatorics \cite{katona}, and Group Testing \cite{DuHw}, to name a few.
It is well known that binary search procedures that proceed with comparison queries are mathematically \textit{equivalent} to binary alphabetic codes, and that  binary search procedures that proceed with membership queries are equivalent to binary prefix codes (see, e.g., \cite{AW}).


Among the numerous versions of the basic binary search problem, 
we are interested in the
following variation that arose in the context of
efficient leader election in asynchronous
anonymous networks \cite{vL+}. We first describe the problem, then we show how it gives rise to
our problem of constructing alphabetic binary 
codes
with  the constraint that each codeword must contain at most a given number of $D$ ones.
We quote  \textit{verbatim} from \cite{vL+},  apart from the 
  names of the parameters

\begin{quote}
    ``Consider the case when the unknown number is a positive integer in the interval
$[1,n]$, and the total amounts $D$ of admissible overestimates is predetermined. The
$\langle n,*,D \rangle$--
game \textit{consists} of determining the unknown number by asking as few questions as
possible''.
\end{quote}

\noindent
The questions are of the type: \lq\lq Is the unknown number greater than $j\in [1,n]$?''.
With the term \textit{overestimate}, the authors of \cite{vL+}  intend  
the situation in which the search algorithm,
 receives a \lq{\sc No}' answer to the query.
In fact, if this is  the case, it means that  the value  $j$ is
\textit{greater} than the value of
the unknown number  
one seeks,  hence,  $j$ is an overestimate of the sought number.
The authors of \cite{vL+}  also  pointed out that their
$\langle n,*,D\rangle$--game (as above defined)
corresponds to the bounded-searching problem with variable-cost 
comparisons studied by Bentley and Brown 
in \cite{BenBro}
and  remarked that their results 
about $\langle n,*,D\rangle$--games
lead to improved solutions for the
algorithmic problems studied in \cite{BenBro}.

In different words, the $\langle n,*,D\rangle$--search game problem 
consists in determining the best
\textit{worst-case} comparison-based 
search algorithm $\cal A$ that determines the unknown 
element in $[1,n]$, under the constraints that $\cal A$ never performs a sequence of comparison with more than $D$ overestimates of the unknown number.
Using the well-known correspondence between ({deterministic}) comparison-based search
procedures and binary trees \cite{mehlhorn}, one can restate  the 
$\langle n,*,D\rangle$--search game problem of \cite{vL+} as the problem of determining 
the \textit{minimum} possible height of any $n$-leaf binary tree, 
 satisfying the  constraints:
\begin{enumerate}
    \item[\bf 1.] any root-to-leaf path contains at most $D$ \textit{right}-pointing edges
    (assuming the  right-pointing edges correspond to the {\sc No} answers),
    \item[\bf 2.]   the leaves of the tree,  read from the leftmost
              to the rightmost leaf, are in one-to-one correspondence
             with the elements of the ordered set $[n]=\{1, \ldots , n\}$.
\end{enumerate}
Additionally, using the equally well-known correspondence between binary trees and
variable-length binary  codes (and assuming to encode each {\sc No} answers with 1, and each each {\sc Yes} answers with 0), one can see the 
$\langle n,*,D\rangle$--search game problem of \cite{vL+} as the problem of determining an \textit{alphabetic} encoding (because of 2. above)  of the elements in $[n]$ such that: \textbf{a)} the encoding has minimum \textit{maximum codeword length},
\textbf{b)} the encoding satisfies the constraint that each
 codeword must contain at most $D$ ones.
In the rest of the paper, we will use interchangeably both  the languages of codes and of trees, since they refer to the same objects, essentially.

A strictly related problem was considered by Dolev \textit{et al.} in \cite{Dolev}.
More precisely, the authors of \cite{Dolev} considered the problem 
of determining 
the {minimum} possible \textit{average} height of any $n$-leaf binary tree
that satisfies the following constraints:
\begin{enumerate}
    \item[\bf 1'.] any
    root-to-leaf path contains at most $D$ \textit{right}--pointing edges,
     \item[\bf 2'.]  the leaves of the tree,  (\textit{not} necessarily read from the leftmost
              to the rightmost leaf), are in one-to-one correspondence
             with the elements of the set $[n]$,
         \item[\bf 3'.] the search space $[n]$ is endowed with a probability
         distribution $p=(p_1, \ldots , p_n)$, where $p_i$ is the probability that 
         the unknown element one wants to determine is equal to $i\in [n].$
    \end{enumerate}
    The authors of \cite{Dolev} were motivated by the goal of
quantifying the time--energy trade-off in message transmissions between mobile hosts
and mobile support stations. 
In the same paper \cite{Dolev} the authors give a $O(n^{2+D})$-time complexity
 algorithm
to construct a \textit{minimum}  \textit{average} height tree
that satisfies the required
properties \textbf{1'}, \textbf{2'}, and \textbf{3'}.
{By the known correspondence between binary trees and search algorithms
that operate by executing \textit{membership queries} (see, e.g. \cite{aigner}, Ch. 1.9)}, one can  see that \textit{any} {minimum}-{average} height binary tree,
satisfying \textbf{1'}, \textbf{2'}, and \textbf{3'}, gives an average-case optimal 
search strategy in which: \textbf{ a)} tests correspond to arbitrary tests of the type:
\lq\lq Does the unknown number belong to a subset $A\subseteq [n]$?''  
and \textbf{b)} the number of tests that can receive a {\sc No}  answer is upper
bounded by $D$. It should be clear, at this point, that the problem can be equivalently formulated
as that of constructing a minimum average length prefix code, in which each
codeword contains at most $D$ ones.

We remark that problems of the kind described above arise in many other different scenarios (e.g., \cite{P+, Sabato+, sabato, Saettler, Saettler2015, Sa, Ser, T+, Tu, YT}) where different test \textit{outcomes} might have different \textit{costs}. 
In our case, a {\sc Yes}
outcome costs 0, while a {\sc No} has some cost $>0$ (like in the classic 
egg-dropping puzzle \cite{K++}, for example), and one seeks average-case optimal
algorithms (with respect to the number of \textit{performed} tests) with a fixed 
bound on the worst-case cost of the algorithm. Related problems are also studied in \cite{GR, KR} and references therein quoted.

\Subsection{Our results}\label{our}
In Section \ref{sec:dynamic} we consider the  scenario of the 
$\langle n,*,D\rangle$--search games
studied in  \cite{vL+}, where now
the search space $[n]$ is endowed with a probability
         distribution $p=(p_1, \ldots , p_n)$, and  $p_i$ is the probability that 
         the unknown element is equal to $i\in [n].$ We design an 
         $O(n^2D)$ time algorithm to find the \textit{best average case}
          comparison-based 
search algorithm $\cal A$ that is capable of determining the unknown 
element in $[n]$, under the constraints that $\cal A$ never performs a sequence of comparisons with more than $D$ overestimates of the unknown number.
Equivalently, our algorithm constructs a 
minimum average height binary tree satisfying 
properties \textbf{1.} and  \textbf{2.} above, for
the input probability
distribution $p=(p_1, \ldots , p_n)$. Our main technical 
tool is  the celebrated  Knuth-Yao  Dynamic Programming--speedup through quadrangle--inequality \cite{knuthp,yao1982speed}.

Successively, in Section \ref{sec:condition} we 
derive the following result.
Let $\{\ell_1,\ldots,\ell_n\}$ be a multiset of positive integers, and $D\geq 1$ be an
integer (by the definition, $\{\ell_1,\ldots,\ell_n\}$ is unordered). We give  a necessary and sufficient 
condition for the {existence} of an $n$-leaf {full} binary tree $T$
(i.e., in which any internal node has two children) 
such that:
\begin{enumerate}
    \item[\bf a)] the multiset of all root-to-leaf path lengths is equal to
    $\{\ell_1,\ldots,\ell_n\}$, 
      \item[\bf b)] any root-to-leaf path of $T$ contains at most $D$ \textit{right}--pointing edges.
\end{enumerate}
The interest in the above result is twofold. First, it represents a
natural extension of the classical 
Kraft inequality (see, eg., Thm. 5.1 of \cite{AW}) that establishes a  necessary and sufficient 
condition for the {existence} of an $n$-leaf {full} binary tree
with multiset of root-to-leaf path lengths  equal to 
    $\{\ell_1,\ldots,\ell_n\}$. 
Secondly, the above result will be instrumental to the derivation of 
an improvement of the result by Dolev \textit{et al.} \cite{Dolev}
that we have discussed before. More precisely, 
in Section \ref{dolev} we design a novel 
algorithm
to construct a \textit{minimum}  \textit{average} height tree
that satisfies 
properties \textbf{1'}, \textbf{2'}, and \textbf{3'};
our algorithm has  $O(n^2D)$ time complexity,
considerably improving on the $O(n^{2+D})$ algorithm presented in \cite{Dolev}.
Note that the parameter $D$ can possibly
assume any value in the interval $[1, \log n]$; therefore, 
our algorithm is the first \textit{polynomial} in $n$ time algorithm to solve 
the problem discussed before.

We conclude  this section by mentioning that our work can also be seen as 
an extension of well-known classical results, in the sense that our first problem 
corresponds to the study of binary alphabetic codes \cite{alp} in 
which all codewords 
have a bounded number $D$ of 1's, and the second problem corresponds to the extension 
of the Huffman encoding \cite{huf} to the same constrained scenario 
just mentioned.
As such, our work  fits into the 
long line of research devoted to the construction of 
minimum average variable-length codes satisfying given structural properties
(see, e.g., \cite{garey1974optimal,Golin3,Golin,wessner1976optimal}).


\Section{Preliminaries}\label{problem}
We find it convenient to 
exploit the known equivalence between 
comparison-based search procedures, and binary alphabetic codes (see \cite{AW,aigner}). 

 Let $S=\{s_1, \ldots, s_n\}=[n]$  be the set of symbols,
ordered according to the standard order.
A binary alphabetic code is a
mapping $w:[n]\longmapsto \{0,1\}^+$, 
enjoying the following two properties:
\begin{enumerate}
    \item the mapping 
 $w:[n]\longmapsto \{0,1\}^+$ is order-preserving,
where the order relation on the set of all
binary strings $\{0,1\}^+$ 
is the  alphabetical order,
\item no codeword $w(s)$ is prefix
of another $w(s')$, for any $s,s'\in [n]$, $s\neq s'$.
\end{enumerate}

In the following sections, we will occasionally see an alphabetic code 
as a binary tree $T$ where each left-pointing edge is labeled by bit 0 and each 
right-pointing edge is labeled  by bit 1.
Each leaf
of the tree corresponds to a different symbol  $s\in S$ and 
the order-preserving property
of alphabetic codes implies that the leaves of $T$, read from the leftmost
leaf to
the rightmost one, appear in the order $1<2<\ldots <n$.

Given a probability distribution $p=(p_1,\dots,p_n)$, where  $p_i$ is
the probability of symbol $s_i$, for $i=1, \ldots , n$, we denote with the symbol  $ \mathbb{E}[C]$ the average  length of an alphabetic code $C=\{w(s): s\in S\}$ for $S$, defined as
  $ \mathbb{E}[C] = \sum_{i=1}^n p_i \ell_i,$
where $\ell_i$ is the length of the codeword $w(s_i)$.

In Section \ref{intr} we interpreted 
the 
$\langle n,*,D\rangle$--search games of \cite{vL+}
in terms of constrained $n$-leaf binary trees.
That said, we can translate our algorithmic problems into the language of alphabetic codes
as follows.
             For any binary string $w \in \{0,1\}^+$, let us denote with $\ones{w}$ the number of 1's that appear in $w$.
      The problem of designing the {best average case}
          comparison-based 
search algorithm $\cal A$ to  determine the unknown 
element in $[n]$, under the constraints that $\cal A$ never performs a sequence of comparison with more than $D$ overestimates of the unknown number, is equivalent
to solving
 the following optimization problem: 

 \vspace*{-.6cm}
\begin{align}\label{eq:problem}
    &\min_{C \mbox{ \footnotesize alphabetic}}\mathbb{E}[C]= 
    \min_{C \mbox{ \footnotesize alphabetic}}\sum_{i=1}^n p_i\,\ell_i,\\
    &\mbox{\ \ \ \ subj. to } \ones{w(s_i)} \leq D, \quad \forall i=1,\dots,n,\nonumber
\end{align}

\vspace*{-.1cm}
\noindent
where $\ell_i$ is the length of $w(s_i)$. In other words, one seeks an alphabetic code of minimum average length in which \textit{each} codeword contains no more than $D$ ones.

\Section{A Dynamic Programming approach}\label{sec:dynamic}
In this section, we first provide a straightforward Dynamic Programming algorithm to compute
a solution to the optimization problem (\ref{eq:problem}).
Subsequently, we employ the  Knuth-Yao  Dynamic Programming--speedup to 
get a better algorithm.

Let $C(i,j, w)$ denote 
the \textit{minimum} average length of any
alphabetic code (tree) for the symbols $s_i,\ldots,s_j$, with $1 \leq i\leq j \leq n$, in which each codeword contains at most $w$ ones.
We say that $C(i,j, w)$ is 
the \textit{cost of the optimal alphabetic tree}.
The solution of
the optimization problem 
 (\ref{eq:problem}) corresponds to the problem of determining the value $C(1,n,D)$. Let 
$
\s(i,j) = p_i+\cdots+p_j
$
denote the sum of the probabilities associated with the symbols $s_i,\ldots,s_j$. 
We  express
the cost $C(i,j,w)$ as a function of the cost of its subtrees.
By considering all the possible ways of splitting the sequence
of symbols $s_i,\ldots,s_j$ into a left and a right subtree (with at least a node in each subtree), we get the following recurrence relation: 
\begin{equation}\label{eq:relazione}
    C(i,j,w) = 
    \begin{cases}
        0 &\!\!\!\!\!\!\!\!\!\!\!\!\!\!\!\!\!\!\!\!\!\!\!\!\!\!\!\!\!\!\!\!\!\!\!\!\!\!\!\!\!\!\!\!\!\!\text{ if } i=j,\\
        \infty &\!\!\!\!\!\!\!\!\!\!\!\!\!\!\!\!\!\!\!\!\!\!\!\!\!\!\!\!\!\!\!\!\!\!\!\!\!\!\!\!\!\!\!\!\!\!\text{ if } i<j \text{ and } w=0,\\
        \s(i,j) + \min_{i < k \leq j}\{C(i,k-1, w) + \\ \qquad C(k, j, w-1)\} &\!\!\!\!\!\!\!\!\!\!\!\!\!\!\!\!\!\!\!\!\!\!\!\!\!\!\!\!\!\!\!\!\!\!\!\!\text{ otherwise.}
    \end{cases}
\end{equation}
We remind that the right subtree is reached from a right-pointing edge of the 
root; therefore, all the corresponding codewords one obtains begin with symbol 1, and this 
accounts for the value $w-1$ in the parameters of the right subtree cost.

From the expression of  (\ref{eq:relazione}),
one can easily see that $C(1,n,D)$, 
can be computed in 
 $O(n^3D)$ time and $O(n^3D)$ space.
We explicitly write the \textbf{Algorithm \ref{alg:naive}} that does that, since it will be useful
for discussing the improvement in the time complexity, presented in the next 
Section \ref{sec:improved}. To compute also 
the optimal alphabetic code (and not only its cost)
the algorithm keeps track  
of the indices on which the \textit{optimal} partition occurs.
This is done with the values $R(i,j,w)$'s, i.e., $R(i,j,w)$ corresponds to the index $k$ for which $
    C(i,j,w) = \s(i,j) + C(i,k-1,w) + C(k,j,w-1).
$
{

\begin{algorithm}
\begin{scriptsize}
\caption{}
\label{alg:naive}
\KwIn{Symbols $S=\{s_1,\dots,s_n\}$, $p=(p_1,\dots,p_n)$ the associated probability distribution and an integer value $D$}

\For{$i \gets 1$ \textbf{to} $n$}{
    \For{$j \gets i+1$ \textbf{to} $n$}{
        $C(i,j,0) = \infty$
    }
}

\For{$w \gets 0$ \textbf{to} $D$} {
    \For{$i \gets 1$ \textbf{to} $n$} {
        $C(i,i,w)=0$
        
        $R(i,i,w)=i$    
    }
}

\For{$w \gets 1$ \textbf{to} $D$} {
    \For{$i \gets 1$ \textbf{to} $n-1$} {
        $C(i,i+1,w)=p_{i} + p_{i+1}$
        
        $R(i,i+1,w)=i$    
    }
}
\For{$w \gets 1$ \textbf{to} $D$} {
    \For{$s \gets 2$ \textbf{to} $n-1$} { 
        \For{$i \gets 1$ \textbf{to} $n-s$} {
    
             $j=i+s$ \tcp*{Proceed by diagonals}
        
             $min=\infty$
             
             $mindex = -1$
             
             \For{$k \gets i+1$ \textbf{to} $j$} {
                \If {$C(i,k, w)+C(k+1,j, w-1)\leq min$} {
                    $min=C(i,k-1, w)+C(k,j, w-1)$
    
                    $mindex=k$
                    }
              }
              $C(i,j, w)=\left (\sum_{k=i}^j p_k\right ) +min$
              
              $R(i,j, w)=mindex$
        }
    }
}
\KwOut{$C(1,n, D)$ and $R(\cdot,\cdot,\cdot)$}
\end{scriptsize}
\end{algorithm}
}
\Subsection{An improved algorithm}\label{sec:improved}
The time complexity of Algorithm \ref{alg:naive} can be improved by a factor of $n$ by
using a technique introduced by Knuth in \cite{knuthp} and generalized by 
Yao \cite{yao1982speed}. 
For such a purpose, let us recall some preliminary definitions and results.

Let $R_{i,j,w}$ be the largest index where the minimum is achieved in the definition of $C(i, j, w)$. We recall that $R_{i,j,w}$ corresponds exactly to the value $R(i,j,w)$ computed in \textbf{Algorithm \ref{alg:naive}}. Moreover, let us use $C_k(i,j,w)$ to denote the cost of an optimal alphabetic code under the constraint that the 
first splitting of symbols $s_i,\ldots,s_j$ 
is performed into $s_i,\dots,s_{k-1}$ and $s_k,\ldots,s_j$, i.e, $C_k(i,j,w)$ satisfies

\vspace*{-.5cm}
\begin{equation}\label{eq:Ck}
    C_k(i,j,w) = \s(i,j) + C(i,k-1,w) + C(k,j,w-1).
\end{equation}

\vspace*{-.2cm}
\noindent

We recall that a two-argument function $f(\cdot, \cdot)$ satisfies the \textit{Quadrangle Inequality} \cite{yao1982speed} if it holds that
   $f(i,j) + f(i',j') \leq f(i',j)+ f(i, j'), \ \forall 
    1\leq i\leq i'\leq j\leq j'\leq n.$
A function $f(\cdot, \cdot)$ is said to be \textit{monotone} if 
it holds that 
$f(i,j') \geq f(i',j), \quad \forall 1\leq i\leq i'\leq j\leq j'\leq n.$
One can easily see that $\s(i,j)=p_i+\cdots+p_j$ is both monotone and satisfies the Quadrangle Inequality (actually, with the equality sign).

We also recall the work of Borchers and Gupta \cite{borchers1994extending} in which they showed that the Quadrangle Inequality holds for a large class of functions. Let $F(i,j,r)$ be the minimum cost of solving a problem on inputs $i,i+1,\dots,j$ with \lq\lq resource" $r$ (in our case,
the resource will be  the  maximum  number of 1's one can have in each codeword).
Borchers and Gupta showed that the Quadrangle Inequality holds for all functions $F(i,j,r)$ defined as follows:

\begin{equation}\label{eq:relazione_F}
    F(i,j,r) = 
    \begin{cases}
        w(i), &\!\!\!\!\!\!\!\!\!\!\!\!\!\!\!\!\!\!\!\!\!\!\!\!\!\!\!\!\!\!\!\!\!\!\!\!\!\!\!\!\!\!\!\!\!\!\!\!\!\!\!\!\!\!\!\!\!\!\!\!\!\!\!\!\!\!\!\text{ if } i=j \text{ and } r\geq 0,\\
        \min_{i \leq k < j} \{aF(i,k,f(r)) + bF(k+1,j,g(r)) \\
        \qquad + h(i,k,j)\}, & \!\!\!\!\!\!\!\!\!\!\!\!\!\!\!\!\!\!\!\!\!\!\!\!\!\!\!\!\!\!\!\!\!\!\!\!\!\!\!\!\!\!\!\!\!\!\!\!\!\!\!\!\!\!\!\!\!\!\!\!\!\!\!\!\!\!\!\text{ if } f(r)\geq 0 \text{ and } g(r)\geq 0,\\
        \mbox{undefined}, & \!\!\!\!\!\!\!\!\!\!\!\!\!\!\!\!\!\!\!\!\!\!\!\!\!\!\!\!\!\!\!\!\!\!\!\!\!\!\!\!\!\!\!\!\!\!\!\!\!\!\!\!\!\!\!\!\!\!\!\!\!\!\!\!\!\!\! \text{ otherwise,}
    \end{cases}
\end{equation}
where $w(i)$ denotes the cost of $F(i,i,r)$, $f(\cdot)$ and $g(\cdot)$ are two functions such that $r\geq f(r)$ and $r\geq g(r)$, i.e., they are assuming that the resource does not increase, $a$ and $b$ are two positive integers, and the function $h(i,k,j)$ is the cost of dividing the problem $F(i,j,r)$ at the splitting point $k$. Moreover, the function $h$ must satisfy the following conditions that are similar to the quadrangle inequality: for $i\leq j \leq t < s \leq l$ and $i\leq s < l$, if $t\leq k$ it holds that 
\begin{equation}\label{cond1}
    h(i,t,s)-h(j,t,s)+h(j,k,l)-h(i,k,l)\leq 0 
\end{equation}
and 
\begin{equation}\label{cond2}
    h(j,k,l)-h(i,k,l) \leq 0.
\end{equation}
While if $k < t$, it holds that 
\begin{equation}\label{cond3}
    h(j,t,l)-h(j,t,s)+h(i,k,s)-h(i,k,l)\leq 0 
\end{equation}
and 
\begin{equation}\label{cond4}
    h(i,k,s)-h(i,k,l) \leq 0.
\end{equation}
The authors of \cite{borchers1994extending} proved the following result.
\begin{lemma}[\cite{borchers1994extending}]\label{lemma_F}
   For each $r\geq 0$,
   the function $F$ defined in (\ref{eq:relazione_F}) satisfies the Quadrangle Inequality, 
   that is, $\forall 1\leq i \leq i'\leq j \leq j' \leq n$ it holds that
     \begin{equation}\label{eq:qi_inequality_F}
         F(i,j,r)+F(i',j',r) \leq F(i,j',r) + F(i',j,r).
     \end{equation}
\end{lemma}

Using  Lemma \ref{lemma_F}
 we can prove the following result.

\begin{corollary}\label{cor:QI}
    For each $0\leq w\leq D$,
   our function $C$ defined in (\ref{eq:relazione}) satisfies the Quadrangle Inequality, 
   that is, 
   $\forall 1\leq i \leq i'\leq j \leq j' \leq n$, it holds that
     \begin{equation}\label{eq:qi_inequality_C}
         C(i,j,w)+C(i',j',w) \leq C(i,j',w) + C(i',j,w).
     \end{equation}
\end{corollary}

\begin{proof}
    We need to show that the function $C$ is a special case of (\ref{eq:relazione_F}). For such a purpose, let $w(i)=0$ for each $i=1,\dots,n$, and let $a=b=1$. While as functions $f$ and $g$ we take $f(x)=x$ and $g(x)=x-1$. Finally, we define the function $h(i,k,j)=\s(i,j)$ for each $i\leq k < j$, where $\s(i,j)=p_i+\dots+p_j$. One can see that the function $h$ defined in this way satisfies the conditions (\ref{cond1}-\ref{cond4}). In fact, for $i\leq j \leq t < s \leq l$ and $i\leq s$, if $t\leq k$ it holds that
    \begin{align}
        h(i,t,s)&-h(j,t,s)+h(j,k,l)-h(i,k,l) \\
                        &= \s(i,s)-\s(j,s)+\s(j,l)-\s(i,l)=0
    \end{align}
    and 
    \begin{align}
        h(j,k,l)-h(i,k,l) = \s(j,l)-\s(i,l) < 0.
    \end{align}
    Similarly, if $k<t$, it holds that
    \begin{align}
        h(j,t,l)&-h(j,t,s)+h(i,k,s)-h(i,k,l) \\
                &= \s(j,l)-\s(j,s)+\s(i,s)-\s(i,l)
        = 0
    \end{align}
    and 
    \begin{align}
        h(i,k,s)-h(i,k,l) = \s(i,s)-\s(i,l) < 0.
    \end{align}
    Since we have shown that the function $C$ is a special case of (\ref{eq:relazione_F}), from Lemma \ref{lemma_F} we get that the function $C$ satisfies the Quadrangle Inequality. 
\end{proof}

By using Corollary \ref{cor:QI} we can show the following result which is a generalization of the 
crucial property called ``monotonicity of the roots" in  \cite{knuthp}. 
This result will be the key step for the improved algorithm.
\begin{theorem}\label{th}
    For each $0\leq w\leq D$ it holds that
    \begin{equation}
       \forall 1\leq i\leq j\leq n \quad  R_{i,j,w}\leq R_{i,j+1,w}\leq R_{i+1,j+1,w}.
    \end{equation}
\end{theorem}

\begin{proof}
    Let $1 \leq w \leq D$ (for $w=0$ the theorem is trivially true). The claim is trivially true also when $i=j$. Let us assume that $i < j$. We will show that $R_{i,j,w} \leq R_{i,j+1,w}$, the argument for $R_{i,j+1,w}\leq R_{i+1,j+1,w}$ follows by symmetry.
    Since $R_{i,j,w}$ is the largest index where the minimum is achieved in the definition of $C(i,j,w)$, it is sufficient to show that $\forall i<k\leq k'\leq j$
    it holds that 
    \begin{align} \label{eq:inequality_1}
        C_{k'}(i,j,w) \leq C_{k}(i,j,w) {\implies} C_{k'}(i,j+1,w) \leq C_{k}(i,j+1,w).
    \end{align}
    We will prove {that (\ref{eq:inequality_1}) holds by showing the following  inequality, for all $i<k\leq k'\leq j,$}   
    \begin{equation}\label{eq:inequality}
        C_k(i,j,w) - C_{k'}(i,j,w) \leq C_k(i,j+1,w) - C_{k'}(i,j+1,w).
    \end{equation}
    {%
    In fact, let us assume that (\ref{eq:inequality}) is true. Then, whenever  $ C_{k'}(i,j,w) \leq C_{k}(i,j,w)$ of (\ref{eq:inequality_1}) holds, 
    from (\ref{eq:inequality}) we get that 
    \begin{equation}\label{eq:nn}
        0\leq C_k(i,j,w) - C_{k'}(i,j,w) \leq C_k(i,j+1,w) - C_{k'}(i,j+1,w).
    \end{equation}
    From (\ref{eq:nn}),  since $0\leq C_k(i,j+1,w) - C_{k'}(i,j+1,w)$, we obtain that $ C_{k'}(i,j+1,w) \leq C_{k}(i,j+1,w)$, that is,  (\ref{eq:inequality_1}) holds too. To prove  (\ref{eq:inequality}) let us rewrite it as follows}
    \begin{equation}\label{rewritten}
         C_k(i,j,w) +C_{k'}(i,j+1,w) \leq C_k(i,j+1,w) +C_{k'}(i,j,w).
    \end{equation}
    By expanding all four terms of (\ref{rewritten}) we get (cfr., see (\ref{eq:Ck}))
    \begin{align*}
        \s&(i,j) + C(i,k-1,w) + C(k,j,w-1) +\s(i,j+1) + C(i,k'-1,w) + C(k',j+1,w-1)\\ 
        &\leq \s(i,j+1) + C(i,k-1, w) + C(k,j+1,w-1)
        + \s(i,j) + C(i,k'-1,w) + C(k',j,w-1).
    \end{align*}
    Grouping the terms of the  above inequality, we obtain
    \begin{equation}
         C(k,j,w-1) + C(k',j+1,w-1) \leq C(k,j+1,w-1) + C(k',j,w-1).
    \end{equation}
This is the quadrangle inequality for 
    the function $C$ with parameter $w-1$,
    with $k\leq k'\leq j < j+1$, and from Corollary \ref{cor:QI} we know that it holds true. Thus, this concludes the proof.
\end{proof}

Due to Theorem \ref{th}, instead of searching from $i+1$ to $j$ 
(line 18 of \textbf{Algorithm 1}), it is
sufficient to search from $R(i,j-1,w)$ through $R(i+1,j,w)$. 
{Therefore  line 18 of \textbf{Algorithm 1} can be
substituted by 
\lq\lq \textbf{for} $k \leftarrow R(i,j-1,w)$ \textbf{to} $R(i+1,j,w)$".}
Exactly as argued by Knuth (see the remark after the Corollary, 
p. 19 of \cite{knuthp}), 
when $C(i,j,w)$ is being calculated with this modification, only $R(i+1,j,w)-R(i,j-1,w)+1$ cases need to be examined.  Summing for fixed $j-i$ gives a telescoping series showing that the total amount of work of each iteration of the loop at line 12 of the modified \textbf{Algorithm \ref{alg:naive}} is proportional to $O(n^2)$, at worst. 
This decreases the total time complexity of the algorithm from $O(n^3D)$ to $O(n^2D)$. 

\Section{A Kraft-like condition for the existence of constrained binary trees}\label{sec:condition}
Let $\{\ell_1,\ldots,\ell_n\}$ be a multiset of positive integers, and  $D$ a positive integer. 
In this section, we derive a necessary and sufficient 
condition for the {existence} of an $n$-leaf {full} binary tree $T$
such that:
\textbf{a)} the (unordered)
multiset of all root-to-leaf path lengths is equal to
    $\{\ell_1,\ldots,\ell_n\}$,
    and  \textbf{b)}
      any root-to-leaf path of $T$ contains at most $D$ \textit{right}--pointing edges.
In addition to its intrinsic interest (indeed,  the result represents a
natural extension of the classical
Kraft inequality (Thm. 5.1 of \cite{AW})),
it will  be essential for the derivation of 
an improvement of the result by Dolev \textit{et al.} \cite{Dolev}
that we have anticipated in Section \ref{our}.
Using the natural correspondence between binary trees and binary prefix codes, we will first derive a sufficient condition for the existence of binary prefix codes 
with the \textit{constraint} that the number of 1's in each codeword is bounded
by a given parameter $D$. We point out that there are many extensions of the Kraft inequality
to situations where the codes have to satisfy additional constraints (besides that of being 
prefix),
see \cite{Bassino, carter, DP, GolinDCC, K+, KH, RS,Vi}. However, none of the 
 papers quoted above deals with the extension we derive in this section.

\begin{lemma}\label{lemma:sufficient}
Let $\{\ell_1,\ldots,\ell_n\}$ be a multiset of positive integers, $D$ a positive integer and $L=\max_{i} \ell_i$. Let $N_j$ count the number of integers 
in {$\{\ell_1,\ldots,\ell_n\}$} that are equal to $j$,
for $j=1,\dots, L$.  Then, there exists a binary prefix code
with  codeword lengths $\ell_1,\ldots,\ell_n$, such that each
codeword does not contain more than $D$ $1$'s, \emph{if} the following inequalities hold:

    \begin{align}\label{eq:condition}
        N_j \leq \sum_{i=0}^{D-1} \binom{j}{i} + \binom{j-1}{D-1} - M_{j+1},\quad \forall j =1,\ldots,L, \ \mbox{ where }
    \end{align}

    \vspace*{-.4cm}
    \begin{equation}\label{eq:M}
        M_{j+1} = 
        \begin{cases}
        0 \quad &\mbox{ if } j\geq L,\\
        \left\lceil\frac{\displaystyle N_{j+1} + M_{j+2}}{\displaystyle 2}\right\rceil &\mbox{ if } 1\leq j <L.
        \end{cases}    
    \end{equation}
\end{lemma}

\begin{proof}
As per established convention, we assume that $\binom{j}{i} = 0$ if $i > j$.
    Let $T$ be the binary tree of maximum depth $L$ such that:
    1) $T$ is full, 2) $T$ has the maximum number of nodes,
    and 3)
    any root-to-leaf path of  $T$  has at most $D$ right-pointing edges
    
   We label each left-pointing edge of $T$ with $0$ and 
    each right-pointing edge with $1$.
    The set of all binary words  one gets by reading the 
    sequences of binary labels from the root of $T$ to each leaf
    gives a binary prefix code $P$ in which $\ones{w} \leq D$, for each $w\in P.$
{We now show that the number of nodes at the level $j$ of such a tree $T$ is exactly equal to 
\begin{equation}\label{eq:num_node_j}
        \sum_{i=0}^{D-1} \binom{j}{i} + \binom{j-1}{D-1}\quad \forall j=1, \ldots , L.
\end{equation}

First of all, we recall that each path in the tree can be seen as a binary string. Therefore, (\ref{eq:num_node_j}) holds because the first term   $\sum_{i=0}^{D-1} \binom{j}{i}$ 
counts the number of binary strings of length $j$,
each containing at most $D-1$ ones. The number of binary strings of length $j$ containing exactly $D$ ones 
that one obtains,
by reading the root-to-leaves path in $T$
is equal to  $\binom{j-1}{D-1}$. This is true because any binary string of length $j$ that one obtains
by reading the root-to-leaves path in $T$ \textit{and} that contains $D$ ones must end with 1.
In fact, let us consider a binary string of length $j$ that contains exactly $D$ ones and does not end with 1. In the path from the root to the node associated with such a binary string, there will be at least one internal node that has only exactly \textit{one} child. Therefore, such a node cannot belong at level $j$ of the binary tree $T$ since, by hypothesis, $T$ is full.
}

We show that for an arbitrary
   list of positive integers $\ell_1,\ldots,\ell_n,D$ that satisfy the inequalities (\ref{eq:condition}), there is a prefix code $C$
   with $\ell_1,\ldots,\ell_n$ as codeword lengths, obeying the
   constraint that
   $\ones{w} \leq D$, for each $w\in C.$ 
   The idea of the proof is to build the code by taking as codewords the binary sequences 
   one gets by following the paths from the root  to 
    nodes of $T$ of levels $\ell_1,\ldots,\ell_n$, 
  starting from level  $L=\max_{i} \ell_i$ (i.e., by constructing the 
  longest codewords first).  From (\ref{eq:condition}) we know that
    \begin{equation}\label{eq:construct}
        N_L \leq \sum_{i=0}^{D-1} \binom{L}{i} + \binom{L-1}{D-1}.
    \end{equation}
    Hence, there are enough available nodes in the level $L$ of $T$ to construct $N_L$ codewords of length $L$. In particular, it will be convenient (for future
    purposes)  to construct such codewords by choosing  the first $N_L$ nodes,
    in sequence, starting from the left-most one. The constraint that the obtained 
    codewords $w$'s have to satisfy (i.e., that $\ones{w} \leq D$) is 
    automatically satisfied since in the tree $T$ 
    any root-to-leaf path has at most $D$ right-pointing edges, each labeled with $1$.
    
    The nodes at level $L-1$ that are parents of the chosen nodes at level $L$ \textit{cannot} be used anymore in the construction of the code,  
    to comply with the prefix constraint of $C$. Therefore, we need to ignore such nodes at level $L-1$ and their number is exactly $M_L = \lceil N_L/2\rceil.$
    Thus, the number of nodes at level $L-1$ that can be still used is 
    \begin{equation}
        \sum_{i=0}^{D-1} \binom{L-1}{i} + \binom{L-2}{D-1} - M_L.
    \end{equation}
    From (\ref{eq:condition}) and (\ref{eq:M}) we know that
    \begin{equation}
        N_{L-1} \leq \sum_{i=0}^{D-1} \binom{L-1}{i} + \binom{L-2}{D-1} - M_L,
    \end{equation}
    therefore we can iterate the same process to construct the 
    desired number of codewords of length $L-1$. 
    We choose the first $N_{L-1}$ available nodes at level $L-1$ 
    (that are all the nodes at level $L-1$ that 
    give rises to not-prefixes of the chosen nodes at level $L$) in sequence starting from the left-most one. Again,  we have to ignore all the nodes at level $L-2$ 
    that are ancestors of the chosen nodes at level $L$ and $L-1$.
    This number is exactly 
    \begin{equation}
        M_{L-1} = \left\lceil\frac{\displaystyle N_{L-1} + M_{L}}{\displaystyle 2}\right\rceil.
    \end{equation}
  From (\ref{eq:condition}) and (\ref{eq:M}) we know that
    \begin{equation}
         N_j \leq \sum_{i=0}^{D-1} \binom{j}{i} + \binom{j-1}{D-1} - M_{j+1},\quad \forall j =1,\ldots,L.
    \end{equation}
    We can iterate this operation for each level of $T$. Thus, we can construct a prefix code with codeword lengths $\ell_1,\ldots,\ell_n$ in which each
    codeword does not contain more than $D$ ones.
\end{proof}

To clarify the procedure for the construction of the prefix code described in the previous Lemma, let us consider a concrete example.
\begin{example}
    Let $\langle4, 2, 3, 4, 3, 3, 3\rangle$ be the list of lengths associated to the symbols $s_1,\ldots,s_7$ and $D=2$. By the definitions,  we get the  values: $N_4 = 2, N_3=4, N_2 = 1$ and $N_1=0$. Let us first of all see that for such values the inequalities (\ref{eq:condition}) hold. Indeed:
    \begin{align}
        N_4 = 2 &\leq \sum_{i=0}^1 \binom{4}{i} + \binom{3}{1} = 8,\label{eq:N_4}\\
        N_3 = 4 &\leq \sum_{i=0}^1 \binom{3}{i} + \binom{2}{1} - \left\lceil \frac{2}{2}\right\rceil = 5,\label{eq:N_3}\\
        N_2=1 &\leq \sum_{i=0}^1 \binom{2}{i} + \binom{1}{1} - \left\lceil \frac{4 + 1}{2}\right\rceil= 1,\label{eq:N_2}\\
        N_1 = 0 &\leq \sum_{i=0}^1 \binom{1}{i} + \binom{0}{1} - \left\lceil \frac{1 + 3}{2}\right\rceil = 0\label{eq:N_1}.
    \end{align}
    Since the inequalities hold, we can construct the prefix code with the required
    list of lengths $\langle4, 2, 3, 4, 3, 3, 3\rangle$. For such a purpose, let us start from the complete binary tree $T$ of depth $4$  in which, for any node $v$,  the path 
    from the root to the node $v$ does not contain more than $D=2$ edge with label $1$ (Figure \ref{fig:T}).
    
    \begin{figure}[H]
        \centering
        \scalebox{.8}{
        \begin{tikzpicture}[level distance=1.5cm, level/.style={sibling distance=100mm/#1},
        level 3/.style={sibling distance=2cm}, 
        level 4/.style={sibling distance=1cm}]]
        \node[circle,draw] (root) {}
            child {node[circle,draw] (l1) {}
                child {node[circle,draw] {}
                    child {node[circle,draw] {}
                        child {node[circle,draw] {}}
                        child {node[circle,draw] {}}
                    }
                    child {node[circle,draw] {}
                        child {node[circle,draw] {}}
                        child {node[circle,draw] {}}
                    }
                }
                child {node[circle,draw] {}
                    child {node[circle,draw] {}
                        child {node[circle,draw] {}}
                        child {node[circle,draw] {}}
                    }
                    child {node[circle,draw] {}}
                }
            }
            child {node[circle,draw] (r1) {}
                child {node[circle,draw] {}
                    child {node[circle,draw] {}
                        child {node[circle,draw] {}}
                        child {node[circle,draw] {}}
                    }
                    child {node[circle,draw] {}}
                }
                child {node[circle,draw] {}}
            };
    \end{tikzpicture}}
    \caption{The complete binary tree $T$ of depth $4$}
    \label{fig:T}
    \end{figure}
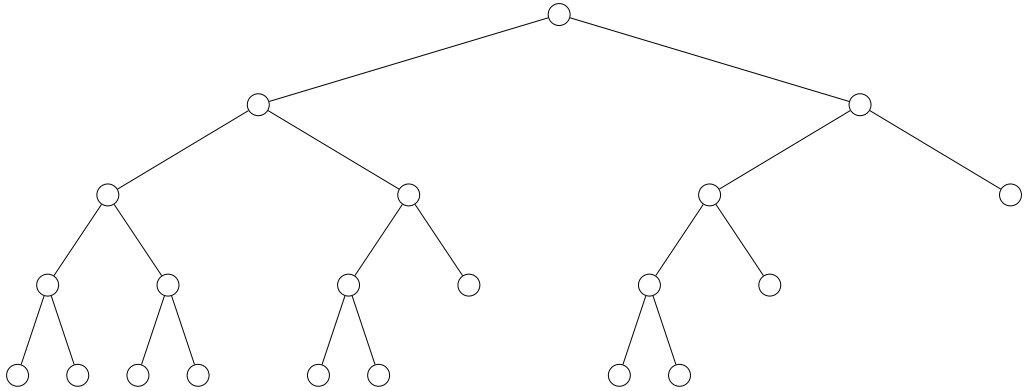
    Since $N_4 = 2$ we need  two nodes at the level $4$ of $T$ to construct the codewords for the symbols $s_1$ and $s_4$, respectively. As described in Lemma \ref{lemma:sufficient}, we take the left-most ones and we get the tree in Figure \ref{fig:T_step_1}. Note that is the inequality (\ref{eq:N_4}) that assures us that we can choose such nodes at the level $4$.
    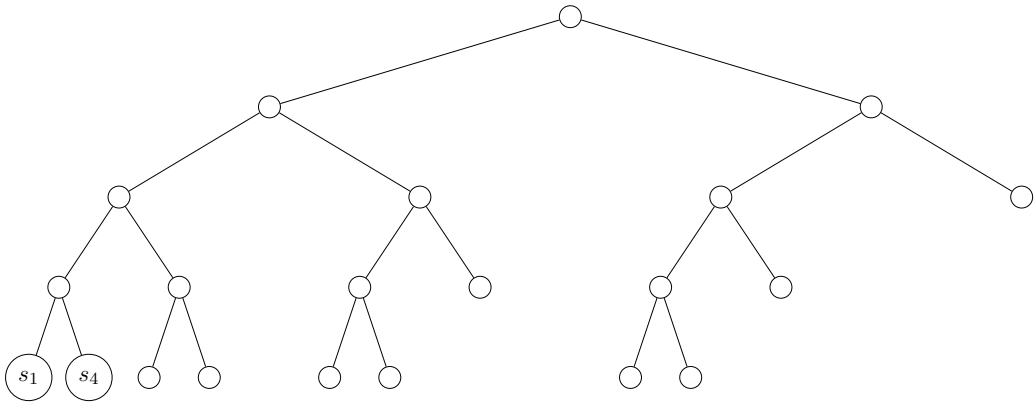
\begin{figure}[H]
        \centering
        \scalebox{.8}{
        \begin{tikzpicture}[level distance=1.5cm, level/.style={sibling distance=100mm/#1},
        level 3/.style={sibling distance=2cm}, 
        level 4/.style={sibling distance=1cm}]]
        \node[circle,draw] (root) {}
            child {node[circle,draw] (l1) {}
                child {node[circle,draw] {}
                    child {node[circle,draw] {}
                        child {node[circle,draw] {$s_1$}}
                        child {node[circle,draw] {$s_4$}}
                    }
                    child {node[circle,draw] {}
                        child {node[circle,draw] {}}
                        child {node[circle,draw] {}}
                    }
                }
                child {node[circle,draw] {}
                    child {node[circle,draw] {}
                        child {node[circle,draw] {}}
                        child {node[circle,draw] {}}
                    }
                    child {node[circle,draw] {}}
                }
            }
            child {node[circle,draw] (r1) {}
                child {node[circle,draw] {}
                    child {node[circle,draw] {}
                        child {node[circle,draw] {}}
                        child {node[circle,draw] {}}
                    }
                    child {node[circle,draw] {}}
                }
                child {node[circle,draw] {}}
            };
    \end{tikzpicture}}
    \caption{The tree $T$ after choosing the nodes for the codewords of length 4}
    \label{fig:T_step_1}
    \end{figure}
    Now, since $N_3=4$ we need to take four nodes at level $3$ for the construction of the codewords for the symbols $s_3,s_5,s_6$ and $s_7$. However, we have to ignore the nodes that are prefixes of the chosen codewords of length $4$. Note that we can do this since the 
    values $N_i$'s satisfy the required inequalities.
    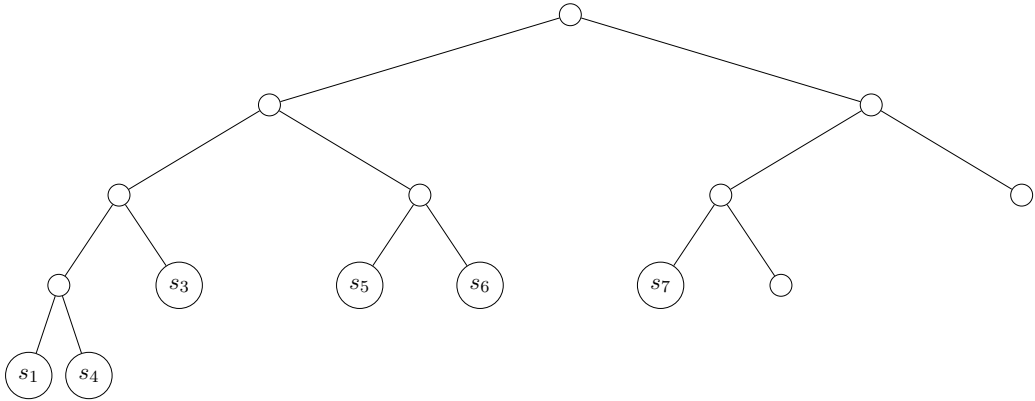
\begin{figure}[H]
        \centering
        \scalebox{.8}{
        \begin{tikzpicture}[level distance=1.5cm, level/.style={sibling distance=100mm/#1},
        level 3/.style={sibling distance=2cm}, 
        level 4/.style={sibling distance=1cm}]]
        \node[circle,draw] (root) {}
            child {node[circle,draw] (l1) {}
                child {node[circle,draw] {}
                    child {node[circle,draw] {}
                        child {node[circle,draw] {$s_1$}}
                        child {node[circle,draw] {$s_4$}}
                    }
                    child {node[circle,draw] {$s_3$}}
                }
                child {node[circle,draw] {}
                    child {node[circle,draw] {$s_5$}}
                    child {node[circle,draw] {$s_6$}}
                }
            }
            child {node[circle,draw] (r1) {}
                child {node[circle,draw] {}
                    child {node[circle,draw] {$s_7$}
                    }
                    child {node[circle,draw] {}}
                }
                child {node[circle,draw] {}}
            };
    \end{tikzpicture}}
    \caption{The tree after choosing the nodes for the codewords of length $4$ and $3$}
    \label{fig:T_ste_2}
    \end{figure}

    Since $N_2=1$ there is only  one codeword of length 2 that we need to choose. 
    Therefore, we take the left-most node at level $2$ that is still ``available" as shown in Figure \ref{fig:step_2}.
    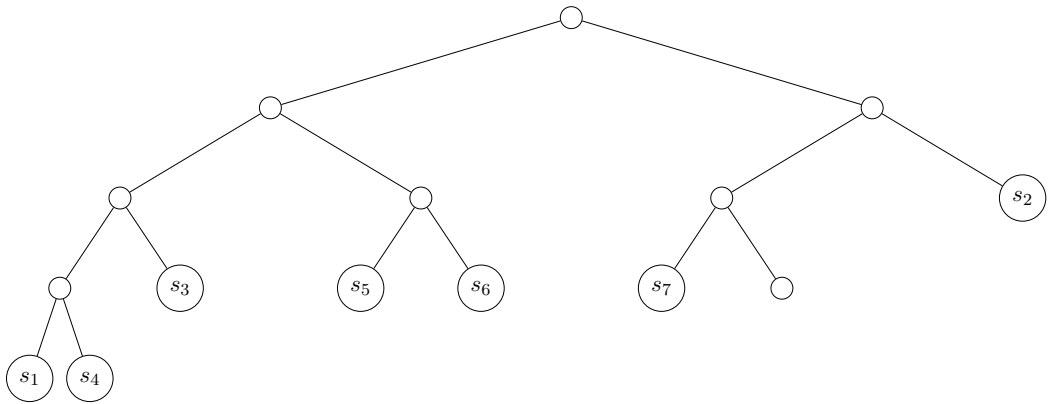
\begin{figure}[H]
        \centering
        \scalebox{.8}{
        \begin{tikzpicture}[level distance=1.5cm, level/.style={sibling distance=100mm/#1},
        level 3/.style={sibling distance=2cm}, 
        level 4/.style={sibling distance=1cm}]]
        \node[circle,draw] (root) {}
            child {node[circle,draw] (l1) {}
                child {node[circle,draw] {}
                    child {node[circle,draw] {}
                        child {node[circle,draw] {$s_1$}}
                        child {node[circle,draw] {$s_4$}}
                    }
                    child {node[circle,draw] {$s_3$}}
                }
                child {node[circle,draw] {}
                    child {node[circle,draw] {$s_5$}}
                    child {node[circle,draw] {$s_6$}}
                }
            }
            child {node[circle,draw] (r1) {}
                child {node[circle,draw] {}
                    child {node[circle,draw] {$s_7$}
                    }
                    child {node[circle,draw] {}}
                }
                child {node[circle,draw] {$s_2$}}
            };
    \end{tikzpicture}}
    \caption{The tree after choosing the nodes for the codewords of length $4, 3$ and $2$}
    \label{fig:step_2}
    \end{figure}
    Finally, since $N_1 = 0$ we do not need to take any node at level $1$. So, the tree that we get in the end is shown in Figure \ref{fig:final}.

    \begin{figure}[H]
        \centering
        \scalebox{.8}{
        \begin{tikzpicture}[level distance=1.5cm, level/.style={sibling distance=100mm/#1},
        level 3/.style={sibling distance=2cm}, 
        level 4/.style={sibling distance=1cm}]]
        \node[circle,draw] (root) {}
            child {node[circle,draw] (l1) {}
                child {node[circle,draw] {}
                    child {node[circle,draw] {}
                        child {node[circle,draw] {$s_1$}}
                        child {node[circle,draw] {$s_4$}}
                    }
                    child {node[circle,draw] {$s_3$}}
                }
                child {node[circle,draw] {}
                    child {node[circle,draw] {$s_5$}}
                    child {node[circle,draw] {$s_6$}}
                }
            }
            child {node[circle,draw] (r1) {}
                child {node[circle,draw] {}
                    child {node[circle,draw] {$s_7$}
                    }
                    child {node[circle,draw] {}}
                }
                child {node[circle,draw] {$s_2$}}
            };
    \end{tikzpicture}}
    \caption{The final tree after the choice of all codewords}
    \label{fig:final}
    \end{figure}
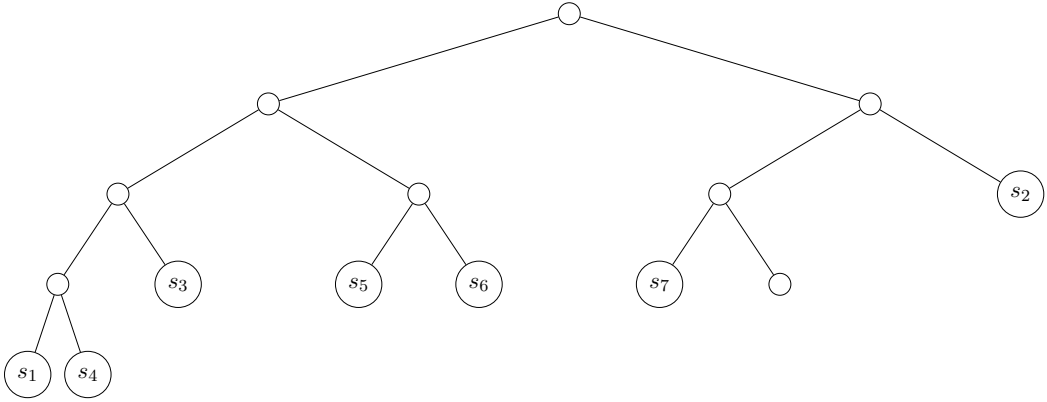
    The prefix code $w: \{s_1,\ldots,s_7\} \to \{0,1\}^+$ that we get from the tree in Figure \ref{fig:final} is the following: $w(s_1) = 0000$, $w(s_2) = 11$, $w(s_3) = 001$, $w(s_4) = 0001$, $w(s_5) = 010$, $w(s_6) = 011$ and $w(s_7)=100$.

    \hfill$\Box$
\end{example}

\medskip

The condition presented in Lemma \ref{lemma:sufficient} 
is  \textit{only} a sufficient condition
for the existence
of binary prefix codes with a given multiset of 
codeword-lengths, satisfying  the {constraint} that the number of ones in each codeword is  bounded
by the parameter $D$.
If we consider multisets containing codeword-lengths 
of  prefix codes whose tree representation corresponds to a {full} binary tree (i.e.,  for which it holds that $\sum_{i=1}^n 2^{-\ell_i} = 1$ or,
equivalently, any internal node has \textit{exactly} two children \cite{AW}), then 
the condition 
of Lemma \ref{lemma:sufficient}
becomes necessary too, as shown in the following Lemma.
\begin{lemma}\label{lemma:necessary}
    Let $\ell_1,\ldots,\ell_n$ be the codeword lengths of a prefix code $C$ in which each
codeword does not contain more than $D$ ones and whose tree representation corresponds to a {full} binary tree, then the following inequalities hold:
\begin{align}\label{eq:condition_2}
        N_j \leq \sum_{i=0}^{D-1} \binom{j}{i} + \binom{j-1}{D-1} - M_{j+1},\quad \forall j =1,\ldots,L,
    \end{align}
    where $M_j$ and $N_j$ are defined as in Lemma \ref{lemma:sufficient} and $L = \max_{i} \ell_i$.
\end{lemma}

\begin{proof}
    {First of all, let us recall that 
    \begin{equation}\label{eq:total}
     \sum_{i=0}^{D-1} \binom{j}{i} + \binom{j-1}{D-1}
    \end{equation}
    is the \textit{maximum} number of nodes that a complete binary tree (in which any root-to-leaf path has at most $D$ right-pointing edges) can have at level $j=1,\dots,L$.
    Moreover, since we are considering only {full} binary trees, one can see that $M_j$, for each $j=2,\dots,L$, corresponds to the \textit{minimum} number of nodes at level $j-1$ that are necessarily prefixes of some nodes associated with codewords of length $\geq j$, regardless of how such nodes are chosen. Putting all together, since by hypothesis, $\ell_1,\dots,\ell_n$ are the codeword lengths of a prefix code $C$ in which each codeword does not contain more than $D$ ones and whose binary tree representation is {full}, we have that the number of leaves at level $j$ in the binary tree of $C$, which is $N_j$, \textit{cannot} be greater than 

    \begin{equation}\label{eq:max}
        \sum_{i=0}^{D-1} \binom{j}{i} + \binom{j-1}{D-1} - M_{j+1},
    \end{equation}
    for each $j=1,\dots,L$. This holds because 
    (\ref{eq:max})
    is obtained by subtracting from the total number
    of nodes at level $j$ (counted by (\ref{eq:total})) the minimum 
    number of nodes at level $j-1$ that are necessarily prefixes of some nodes associated with codewords of length $\geq j$ (counted by $M_{j+1}$).
    In other words, (\ref{eq:max})
    represents the maximum number of leaves at level $j=1,\dots,L$,
    that a complete binary tree --- where any root-to-leaf path has at most $D$ right-pointing edges --- has still \lq\lq available" (i.e., they are not prefixes of the nodes associated with codewords of length $\geq j$). 
    }
\end{proof}

\Section{Optimal Average-Length Constrained Trees}\label{dolev}
In this section, we turn our attention to the problem 
 considered by Dolev \textit{et al.} in \cite{Dolev};  we 
 recall it here for the reader's convenience.
The problem consists in determining 
the {minimum} possible \textit{average} height of any $n$-leaf binary tree
that satisfies the following constraints:
\textbf{a)} any
    root-to-leaf path contains at most $D$ \textit{right}--pointing edges,
\textbf{b)} the leaves of the tree,  (\textit{not} necessarily read from the leftmost
              to the rightmost leaf), are in one-to-one correspondence
             with the elements of the set $[n]$,
\textbf{c)} the  set  $[n]$ is endowed with a probability
         distribution $p=(p_1, \ldots , p_n)$.
By using the known equivalence between binary prefix codes and binary trees, the 
problem is equivalent to the following optimization problem: 
\begin{align}\label{eq:problemDolev}
    \min_{C \mbox{ \footnotesize prefix}}\mathbb{E}[C]&= 
    \min_{C \mbox{ \footnotesize prefix}}\sum_{i=1}^n p_i\,\ell_i,\\
    \mbox{subj. to } \ones{w(s_i)} \leq D, &\quad \forall i=1,\dots,n,\nonumber
\end{align}
where $\ell_i$ is the length of $w(s_i)$. 
We shall prove that the problem (\ref{eq:problemDolev}) 
can be solved in $O(n^2D)$ time, improving on the $O(n^{2+D})$ time algorithm given in \cite{Dolev}. 

The essence of our proof consists in showing that, under the \textit{special} 
condition that the probability distribution $p$ is ordered in a non-decreasing fashion, 
then 
the value of an  optimal solution to the problem (\ref{eq:problemDolev})  
is equal to the value of an optimal solution to (\ref{eq:problem}).

\begin{theorem}\label{th:opt}
Let $S=\{s_1,\ldots,s_n\}$ be a set of symbols, $p=(p_1,\dots,p_n)$ be
    a probability distribution on $S$ such that $p_1\leq\ldots \leq p_n$, and let
    $D$ be a strictly positive integer.
    Then, there exists an optimal prefix code $C$ for $S$ and $p$ such that:
    \textbf{1)} each codeword of $C$ does not contain more than $D$ 1's,
     \textbf{2)}  $\ell_1\geq\ldots\geq \ell_n$, where $\ell_i$ is the length of $w(s_i)$, and 
     \textbf{3)}
 in the tree representation of $C$ the leaves associated to the symbols $s_1,\ldots,s_n$ appear consecutively, from left to right.  
\end{theorem}

\begin{proof}
The existence of an optimal (i.e., with minimum average length) prefix code $C'$ 
in which the codeword lengths are ordered in opposite fashion with respect to the
probabilities ordering is obvious. Therefore, properties \textbf{1)} and \textbf{2)}
are true.
  Moreover, since $C'$ is an optimal code, its tree representation is a {full} binary tree. Therefore, from Lemma \ref{lemma:necessary} we have that the 
  codeword lengths $\ell_1,\ldots,\ell_n$ of $C'$ satisfy the inequalities
   \begin{equation}\label{eq:inequalities}
        N_j \leq \sum_{i=0}^{D-1} \binom{j}{i} + \binom{j-1}{D-1} - M_{j+1},\quad \forall j =1,\ldots,\ell_1.
   \end{equation}
   On the other hand, from Lemma \ref{lemma:sufficient} we know that there exists a prefix code $C$ for $S$ with codewords lengths $\ell_1,\ldots,\ell_n$ such that in the tree representation of $C$  the leaves associated to the symbols $s_1,\ldots,s_n$ appear \textit{consecutively} from left to right (for more details, please see the argument right after (\ref{eq:construct}) in the proof of Lemma \ref{lemma:sufficient}). More explicitly, in the constructive proof of Lemma \ref{lemma:sufficient} we build the code tree $T$ by always taking the leftmost available nodes, starting from the lowest level suitable node, and proceeding iteratively according to the decreasing order of 
    the $\ell_1,\ldots,\ell_n$. This concludes the proof.
\end{proof}

The \textit{crucial}  consequence of Theorem \ref{th:opt} is that one can proceed as follows
to solve the optimization problem (\ref{eq:problemDolev}).
First, one observes that the value of a solution to (\ref{eq:problemDolev}) 
\textit{does not} depend
on the componentwise ordering of the probability distribution 
$p=(p_1,\dots,p_n)$ (this is analogous  to the classical 
Huffman problem to find a minimum average length prefix code,
for example). Therefore, we can assume that 
$p_1\leq\ldots \leq p_n$. Under this assumption, Theorem \ref{th:opt}
asserts that there exists a binary tree $T$
representing an \textit{optimal} solution to (\ref{eq:problemDolev})
in which the symbols $s_1, \ldots, s_n$ (that is, ordered according to decreasing probabilities) appear on the leaves of $T$ consecutively, from 
left to right. 
Said equivalently, $T$ represents an \textit{optimal} alphabetic code for the (now ordered)
set of symbols $\{s_1, \ldots, s_n\}$. We have shown in Section \ref{sec:improved}
that we can find an optimal alphabetic code in time $O(n^2D)$, for \textit{any} arbitrary
(but otherwise \textit{fixed}) ordering on the set of symbols $\{s_1, \ldots, s_n\}$.
All together, this shows that an optimal   solution to the  problem (\ref{eq:problemDolev})
can be found in time $O(n^2D)$ by first ordering the 
probability distribution and then running  the 
algorithm presented Section \ref{sec:improved}. 

\Subsection{When the constraint is really a constraint?}\label{sec:opt1:constraint}

The constraint that \textit{at most} $D$ ones must appear in each
 codeword might prevent the construction of the (otherwise unrestricted) optimal tree. However, if $D$ is large enough, the best tree satisfying the constraint can be the same as the optimal (unrestricted) one.
For the case of alphabetic trees establishing the smallest value of $D$ that does not prevent
the construction of the optimal tree is straightforward. Indeed, it is easy to identify probability
distributions (for example, a decreasing dyadic distribution) for which \textit{any} 
optimal alphabetic code 
{must} have a codeword with $n-1$ ones. Hence, any value of $D$ strictly less than $n-1$ is
a constraint that \textit{prevents} the construction of a tree with the same cost of an \textit{optimal
unrestricted} one.

For the case of prefix codes, establishing a similar result is not immediate. In the next theorem
we prove that the smallest value of $D$ that does not prevent the construction of a tree whose
cost is as the cost of an optimal unrestricted tree is $D=\log_2 n$.

\begin{theorem}
    Let $S=\{s_1,\ldots,s_n\}$ be a set of symbols and $p=(p_1,\dots,p_n)$ a probability distribution on $S$ with $p_1\leq \dots \leq p_n$. There exists an optimal prefix tree $T^*$ for $S$, in which each codeword does not contain more than $\log_2 n$ ones.
\end{theorem}
\begin{proof}
    We prove the theorem by induction.

    \smallskip
    \noindent
    {\em Base case.} For $n\leq 3$ one can see that the theorem holds, as shown in Figure \ref{fig:cases}.
    
    \begin{figure}[H]
      \centering
      \begin{subfigure}[b]{0.4\textwidth}
        \centering
        \scalebox{.8}{
        \begin{tikzpicture}[level distance=0.7cm, level/.style={sibling distance=10mm/#1}]
        \node[circle,draw] (root) {$s_1$};
        \end{tikzpicture}}
        \caption{Optimal prefix tree for $n=1$}
        
      \end{subfigure}
      \hfill
      \begin{subfigure}[b]{0.4\textwidth}
        \centering
        \scalebox{.8}{
        \begin{tikzpicture}[level distance=0.7cm, level/.style={sibling distance=10mm/#1}]
        \node[circle,draw] (root) {}
            child { node[circle,draw]{
                $s_1$
            }
            }
            child { node[circle,draw]{
                $s_2$
            }};
        \end{tikzpicture}}
        \caption{Optimal prefix tree for $n=2$}
    \hfill
      \end{subfigure}
      \begin{subfigure}[b]{0.4\textwidth}
        \centering
        \scalebox{.8}{
        \begin{tikzpicture}[level distance=0.7cm, level/.style={sibling distance=30mm/#1}]
        \node[circle,draw] (root) {}
            child { 
                node[circle,draw] {}
                child {
                    node[circle,draw]{$s_1$}
                }
                child {
                    node[circle,draw]{$s_2$}
                }
            }
            child { node[circle,draw]{
                $s_3$
            }};
        \end{tikzpicture}}
        \caption{Optimal prefix tree for $n=3$}
        
      \end{subfigure}
      \caption{}
      \label{fig:cases}
    \end{figure}

\noindent
{\em Inductive step.} Let $T^*$ be an optimal prefix tree for $S$ with $\ell_1\geq\ldots\geq 
\ell_n$, where
\begin{itemize}
    \item $\ell_i$ is the length of the codeword  $w(s_i)$ associated to 
symbol $s_i$, 
\item the leaves of $T^*$, associated to the symbols 
$s_1,\ldots,s_n$, appear consecutively, from left to right. 
\end{itemize} 
We note that such a tree surely exists. Indeed,  
 when the source symbol probability distribution 
 $p=(p_1,\dots,p_n)$  is ordered, an optimal alphabetic code for 
$S$ is also an optimal prefix code for $S$. Let $T_1$ and $T_2$ be the left and the right subtree 
of the root of $T^*$. By the inductive hypothesis, we have that each codeword of $T_1$ does not 
contain more than $\log_2 |T_1|$ ones, where $|T_1|$ is the number of leaves of $T_1$. Similarly, 
each codeword of $T_2$ does not contain more than $\log_2 |T_2|$ ones, where $|T_2|$ is the number 
of leaves of $T_2$. Moreover, since $T^*$ is full, and since $\ell_1\geq \dots \geq \ell_n$ and 
the leaves associated with the symbols $s_1,\ldots,s_n$ appear consecutively, from left to right, 
one gets that $|T_1|\geq |T_2|$. In fact, let us assume by contradiction that $|T_1|<|T_2|$. Let 
$x$ be the depth of the highest leaf in $T_1$ and let $y$ be the depth of the lowest leaf
in $T_2$. Since the trees are full, we have that $x\leq \log_2 |T_1|$ and that
$y\geq \log_2 |T_2|$. From this, we get that

$$
    y\geq \log_2 |T_2| > \log_2 |T_1| \geq x.
$$
However, $y>x$ leads to a contradiction since $\ell_1\geq\dots\geq \ell_n$.

Putting it all together, we get that each codeword of $T^*$ does not contain a number of
1's greater than 

\begin{align*}
        \max\{\log_2 |T_1|, \log_2 |T_2|+1\} &\leq \max\{\log_2 |T_1|, \log_2 \frac{n}{2} + 1\}\quad\mbox{(since $|T_2|\leq|T_1|$)}\\
        &\leq \log_2 n\quad\mbox{(since $|T_1|+|T_2|=n)$}
\end{align*}

This concludes the inductive step and thus the proof.
\end{proof}

\Section{Concluding Remarks}

In this paper, we studied problems that can be seen as search problems in which 
\textit{different} test outcomes (for the \textit{same} test) can have different costs. Specifically, 
a '{\sc No}' test outcome has a cost of 1, while a '{\sc Yes}' test outcome has a 
cost of 0. We sought search strategies with the minimum average number of tests, 
under the constraint that the employed algorithm never incurs a total cost more 
than a given parameter $D$ for any single execution on input instances.
We pointed out several scenarios where our problems naturally occur. 
\remove{
There are many other situations where our framework applies. For instance, 
in comparison-based search procedures, a test that overestimates the unknown number 
might correspond to a destructive --- and therefore, costly --- test
(consider an algorithm that finds the rupture threshold of a material, for example). 
Membership-based search procedures with variable test-outcome costs also arise in 
\cite{P+,T+}. In that scenario, the presence (or absence) of the sought 
element $x$ in the tested subset $A$ of the search space $S$ might cause the release 
of hazardous material as a consequence of the test on $A$. 
In \cite{Elser} it is considered the problem where one uses 
 \textit{caviae} to test for toxicity of a sample. In the case no toxicity is found, a cavia can be "reused", as opposed to the case where toxicity is present, leading to the death of the cavia. 
 
 In general, 
different test outcomes (for the same test) can have different costs and the problem
of designing a search procedure that never incurs a total cost more than a given 
parameter $D$ naturally arises. }

In this paper, we have provided efficient algorithms for the stated problems, 
extending and considerably improving on the known results in the literature. 
An interesting line of research is to develop our approach in such a way as 
to handle arbitrary variable test outcome costs (i.e., not necessarily 0/1). 
While the extension of the results of Section \ref{sec:dynamic} is  feasible, 
extending the results of Section \ref{dolev} to the more general scenario is much
more challenging. The source of difficulties seems to lie in the extension of 
the combinatorial analysis carried out in Section \ref{sec:condition} to the general 
variable test outcome costs. We leave this extension to future investigations.




\end{document}